\documentclass[runningheads]{llncs}
\usepackage[T1]{fontenc}
\usepackage[justification=centering,font=small,labelfont=bf]{caption}
\usepackage{subcaption}
\usepackage{hyperref}
\usepackage{amssymb,amsmath,amsfonts}
\usepackage{color}

\usepackage{tikz}
\usetikzlibrary{calc,hobby}

\definecolor{utsalmon}{RGB}{235, 202, 184} 
\definecolor{utmagenta}{RGB}{207, 0, 114}
\definecolor{utpurple}{RGB}{79, 45, 127}
\definecolor{utnavy}{RGB}{0, 44, 95}
\definecolor{utblue}{RGB}{0, 148, 179}
\definecolor{utlightblue}{RGB}{99, 177, 229}
\definecolor{utforest}{RGB}{0, 103, 90}
\definecolor{utgreen}{RGB}{52, 178, 53}
\definecolor{utgold}{RGB}{136, 123, 27}
\definecolor{utyellow}{RGB}{254, 209, 0}
\definecolor{utorange}{RGB}{236, 122, 8}
\definecolor{utred}{RGB}{198, 12, 48}
\definecolor{utruby}{RGB}{130, 36, 51}
\definecolor{utwarmgrey}{RGB}{81, 60, 64}
\definecolor{utgrey}{RGB}{97, 99, 101}
\definecolor{utcoolgrey}{RGB}{173, 175, 175}

\newcommand{\alg}{\texttt{\textup{SAM}}}
\newcommand{\sam}{\alg}

\def\Exp{{\mathbb{E}}}
\newcommand{\E}[1]{\mbox{$\Exp[\,#1\,]$}}

\newcommand{\opt}{\texttt{\textup{OPT}}}
\newcommand{\eps}{\varepsilon}

\newcommand{\ind}{\mbox{$\mathcal{I}$}}

\newcommand{\R}{\mathbb{R}}
\newcommand{\sw}[1]{\mbox{$\tilde{w}(#1)$}}
\newcommand{\w}{\tilde{w}}
\newcommand{\talg}{T^{\texttt{\textup{ALG}}}}
\newcommand{\tsam}{T^{\alg}}
\newcommand{\trees}{\mathbb{T}}
\newcommand{\llamb}{\boldsymbol{\lambda}}
\newcommand{\ch}[1]{\textup{\textsc{Ch}}(#1)}
\newcommand{\ove}{\overline{e}}

\begin{document}
\title{Stochastic Minimum Spanning Trees \texorpdfstring{\\}{} with a Single Sample}
\author{Ruben Hoeksma \and Gavin Speek \and Marc Uetz } 

\authorrunning{R.\ Hoeksma, G.\ Speek, and M.\ Uetz}
\institute{University of Twente, Mathematics of Operations Research\\
\email{\{r.p.hoeksma,m.uetz\}@utwente.nl\\
g.w.speek@alumnus.utwente.nl}}
\maketitle           

\begin{abstract}
We consider the minimum spanning tree problem in a setting where the edge weights are stochastic from unknown distributions, and the only available information is a single sample of each edge's weight distribution. In this setting, we analyze the expected performance of the algorithm that outputs a minimum spanning tree for the sampled weights. We compare to the optimal solution when the distributions are known. For every graph with weights that are exponentially distributed, we show that the sampling based algorithm has a performance guarantee that is equal to the size of the largest bond in the graph. Furthermore, we show that for every graph this performance guarantee is tight. The proof is based on two separate inductive arguments via edge contractions, which can be interpreted as reducing the spanning tree problem to a stochastic item selection problem. We also generalize these results to arbitrary matroids, where the performance guarantee is equal to the size of the largest co-circuit of the matroid.

\keywords{Minimum Spanning Tree \and Approximation  \and Sampling \and Matroid}
\end{abstract}

\section{Introduction, motivation, and model}
\label{sec:intro}
The paper considers the minimum spanning tree problem, where given an edge-weighted (multi)graph, one is asked to compute a spanning tree of the graph with minimum total weight. It is folklore in combinatorial optimization that there are efficient algorithms to do that \cite{Kruskal1956,Prim1957}. 
Taking inspiration from some earlier papers on modeling and analyzing minimum spanning tree problems under uncertainty~\cite{Ishii1981,Kulkarni1988,Jain1988,HutsonShier2005,DhamdhereEtAl2005,HoffmannEtAl2008,KamousiSuri2011,FockeEtAt2020}, we make the assumption that the edge weights are not known in advance, but come from unknown distributions. 
The goal is to compute a spanning tree that has minimal weight in expectation. 
Inspired by recent work on a stochastic scheduling problem~\cite{PuckUetz24}, we ask if a provably good solution can be obtained under the restrictive information regime where, per edge, only a single sample from the distribution of the edge weight is known. 
This question is also inspired by the recently increasing interest in learning augmented algorithms in combinatorial optimization \cite{alg-predictions-website}, 
combined with the fact that for some stochastic selection problems, e.g., the prophet inequality setting, a single sample is just as good as full distributional knowledge~\cite{RWW2019}. 
Moreover, from a more practical viewpoint it is not unrealistic to assume that the information being communicated is just a single parameter per edge, while in reality, the true weight is governed by an unknown distribution. 
Our contribution is to shed some light on the resulting effects.

Arguably, under the restricted information regime of a single sample, there is only one sensible algorithm, namely to ``follow the samples'' and compute a minimum spanning tree based on the sampled edge weights. 
Our main result is Theorem~\ref{main_thm_graphs}, which shows that, for exponentially distributed edge weights, the expected performance of this sampling based algorithm is bounded. Specifically, it is a $b$-approximation, where $b$ is the size of the largest bond in the given graph. 
Here, the adversary is the so called non-adaptive optimum, which is the optimal solution that has access to the full distributional information on the edge weights, yet just as the algorithm, the adversary does not know their actual realizations. Moreover, for any given graph $G$, we show that there are rate parameters for the exponential distributions of $G$'s edge weights which make this performance bound asymptotically tight.  

Regarding related work, we note that there is quite an amount of papers that address different versions for stochastic minimum spanning tree problems. We give non-exhaustive pointers to the different models that have been considered in the literature, and note that none of these earlier papers has addressed the single sample question that we address here. The models that have been studied are: 
chance constrained optimization for normally distributed edge weights~\cite{Ishii1981}; 
computing an a-priori spanning tree in a regime where only subsets of vertices need to be visited with certain probabilities~\cite{BertsimasEtAl1990};
computing a robust spanning tree for edge weights in intervals~\cite{AronvH2004}; 
two stage models to minimize total expected costs over the two stages~\cite{DhamdhereEtAl2005}; 
analyzing the distribution function of pointwise optimal minimum spanning trees with stochastic weights~\cite{Kulkarni1988,Jain1988,HutsonShier2005}; 
analyzing models where uncertain edge weights can be queried at a cost, and with the goal to minimize the amount or total cost of queries~\cite{HoffmannEtAl2008,FockeEtAt2020}; 
computing the expected weight of spanning trees when edges may or may not be present with independent probabilities \cite{KamousiSuri2011}, or the PAC learnability of minimum spanning trees \cite{EberleEtAl2022}.

The organization of this paper is as follows. In Section~\ref{sec:notation}, we introduce the used notation, and we give a few preparatory technical lemmas on contractions and bonds in Section~\ref{sec:tech_lemmas}. 
Section~\ref{sec:model} introduces the single sample model for stochastic minimum spanning trees.  Our main result for exponentially distributed edge weights follows in Section~\ref{sec:main_results}. Finally, Section~\ref{sec:main_matroids} generalizes the results to matroids.

\section{Notation: Graphs, bonds, minimum spanning trees}
\label{sec:notation}
We consider loopless (multi)graphs $G = (V,E)$ with vertex set $V$ and edge set $E = (e_1, e_2,\dots,e_m)$. Thus, $|E| = m$. A \textit{cut set} in $G$ is a subset of edges whose removal increases the number of connected components of $G$. 
An inclusion-wise minimal cut set is a \textit{bond}; see Fig.~\ref{fig_bond}. 
Note that the removal of a bond of the graph increases the number of connected components by exactly one. 

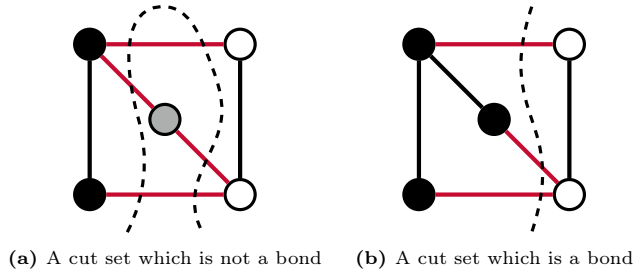
\begin{figure}[ht]
    \centering
    \begin{subfigure}{0.35\textwidth}
	\centering
        \begin{tikzpicture}[node distance=1cm, label distance=-.1cm, very thick] 
        \tikzstyle{vertex} = [circle, minimum size=4mm, inner sep=0mm,draw]
    
        \coordinate (c00) at (0,0);
        \coordinate (c20) at ($(c00)+(2,0)$);
        \coordinate (c02) at ($(c00)+(0,2)$);
        \coordinate (c22) at ($(c00)+(2,2)$);
        \coordinate (c11) at ($(c00)!1/2!(c22)$);
    
        \node[vertex, fill=black] (v02) 
        at (c02) {}; 
        \node[vertex, fill=utcoolgrey] (v11) 
        at (c11) {};
        \node[vertex, fill=black] (v00) 
        at (c00) {};
        \node[vertex, fill=white] (v22) 
        at (c22) {};
        \node[vertex, fill=white] (v20) 
        at (c20) {};
    
        \path[-, ultra thick] 
        (v02) edge[utred] (v22)
        (v02) edge (v00)
        (v02) edge[utred] (v11)
        (v20) edge (v22)
        (v20) edge[utred] (v00)
        (v20) edge[utred] (v11)
        ;
    
        \draw[dashed] (.5,-.5) to[curve through =
        { (.65,1) (1,2.5) (1.55,1.9) (1.75,1) (1.4,0)}] (1.5,-.5);
    
        \end{tikzpicture}
	\caption{A cut set which is not a bond}
    \end{subfigure}
    \begin{subfigure}{0.35\textwidth}
    \centering  
        \begin{tikzpicture}[node distance=1cm, label distance=-.1cm, very thick] 
        \tikzstyle{vertex} = [circle, minimum size=4mm, inner sep=0mm,draw]
        
        \coordinate (c00) at (0,0);
        \coordinate (c20) at ($(c00)+(2,0)$);
        \coordinate (c02) at ($(c00)+(0,2)$);
        \coordinate (c22) at ($(c00)+(2,2)$);
        \coordinate (c11) at ($(c00)!1/2!(c22)$);
        
        \node[vertex, fill=black] (v02) 
        at (c02) {}; 
        \node[vertex, fill=black] (v11) 
        at (c11) {};
        \node[vertex, fill=black] (v00) 
        at (c00) {};
        \node[vertex, fill=white] (v22) 
        at (c22) {};
        \node[vertex, fill=white] (v20) 
        at (c20) {};
        
        \path[-, ultra thick] 
        (v02) edge[utred] (v22)
        (v02) edge (v00)
        (v02) edge (v11)
        (v20) edge (v22)
        (v20) edge[utred] (v00)
        (v20) edge[utred] (v11)
        ;
        
        \draw[dashed] (1.5,2.5) .. controls (1,1) and (2,1) .. (1.5,-.5);
        
        \end{tikzpicture}
    \caption{A cut set which is a bond}
    \end{subfigure}
    \caption{Illustration of cut sets and bonds. 
    Red edges indicate the cut sets. Differently colored vertices are different components 
    after removing the cut set.}
    \label{fig_bond}
\end{figure}
Denote by $B(G)$ a maximum cardinality bond in graph $G$, and note that it need not be unique. 
A (simple) path of length $k$ in $G$ is a set of distinct edges $\{u_i,v_i\}\in E$ so that $v_i=u_{i+1}$ for all $i=1,\dots,k-1$. A cycle of length $k$ in $G$ is a set of $k$ distinct edges $\{u_i,v_i\}\in E$ so that $v_i=u_{i+1}$ for all $i=1,\dots,k$ (mod $k$). A spanning tree $T\subseteq E$ of $G$ is a set of $|V|-1$ edges that does not contain a cycle. Denote by $\trees$ the set of all spanning trees of $G$. If the graph $G$ has edge weights $w:E\to\R$, a minimum spanning tree (MST) is a spanning tree $T\in\trees$ minimizing $w(T)=\sum_{e\in T} w(e)$; it can be computed efficiently, e.g., using Kruskal's (greedy) algorithm~\cite{Kruskal1956} or Prim's algorithm~\cite{Prim1957}.

One operation which will be crucial later is that of an \emph{edge contraction}. If we contract edge $e=\{u,v\}$, vertices $u$ and $v$ are merged into one new vertex $w$, while edge $e$ is removed.  Fig.~\ref{fig:cont} illustrates this.  If the contracted edge $\{u,v\}$ has parallel edges, we delete the resulting loops $\{w,w\}$ incident to $w$, because they are irrelevant the context of spanning trees.
By $G/e$ we denote the graph which is obtained by contracting edge $e$ from $G$. Moreover, by $G\setminus e$, we denote the graph obtained from $G$ by deleting edge $e$.

\begin{figure}[htb]
    \centering
        \begin{tikzpicture}[node distance=1cm] 
        \tikzstyle{vertex} = [circle, minimum size=2.5mm, inner sep=0mm]
    
        \coordinate (cv) at (0,0);
        \coordinate (cu) at ($(cv)+(0,1.5)$);
        \coordinate (cw) at ($(cv)!1/2!(cu)$);
        \coordinate (cu+v) at ($(cu)+(-1,-.5)$);
        \coordinate (cu1) at ($(cu)+(-.7,.5)$);
        \coordinate (cu2) at ($(cu)+(.7,.6)$);
        \coordinate (cu3) at ($(cu)+(0.75,-.4)$);
        \coordinate (cv1) at ($(cv)+(-.7,-.4)$);
        \coordinate (cv2) at ($(cv)+(0.7,-.4)$);
        \coordinate (cv3) at ($(cv)+(0.75,.4)$);
    
        \node[vertex, label=west:{\scriptsize $v$}, fill=utblue] (v) 
        at (cv) {}; 
        \node[vertex, label=west:{\scriptsize $u$}, fill=utblue] (u) 
        at (cu) {};
        \node[vertex, fill=black] (u+v) 
        at (cu+v) {};
        \node[vertex, fill=black] (v1) 
        at (cv1) {};
        \node[vertex, fill=black] (v2) 
        at (cv2) {};
        \node[vertex, fill=black] (v3) 
        at (cv3) {};
        \node[vertex, fill=black] (u1) 
        at (cu1) {};
        \node[vertex, fill=black] (u2) 
        at (cu2) {};
        \node[vertex, fill=black] (u3) 
        at (cu3) {};
    
        \path[-,very thick] 
        (u) edge[utblue] (v)
        (v) edge (v1)
        (v) edge (v2)
        (v) edge (v3)
        (v) edge (u+v)
        (u) edge (u1)
        (u) edge (u2)
        (u) edge (u3)
        (u) edge (u+v)
        (v2) edge (v1)
        (v2) edge (v3)
        (u1) edge (u2)
        (u1) edge (u+v)
        ;

        \draw[-stealth, ultra thick, utred] ($(cw)+(1,0)$) -- ($(cw)+(2,0)$);

        \coordinate (cv2) at (3.25,0);
        \coordinate (cu2) at ($(cv2)+(0,1.5)$);
        \coordinate (cw2) at ($(cv2)!1/2!(cu2)$);
        \coordinate (cu+v2) at ($(cu2)+(-1,-.5)$);
        \coordinate (cu12) at ($(cu2)+(-.7,.5)$);
        \coordinate (cu22) at ($(cu2)+(.7,.6)$);
        \coordinate (cu32) at ($(cu2)+(0.75,-.4)$);
        \coordinate (cv12) at ($(cv2)+(-.7,-.4)$);
        \coordinate (cv22) at ($(cv2)+(0.7,-.4)$);
        \coordinate (cv32) at ($(cv2)+(0.75,.4)$);
    
        \node[vertex, label=north:{\scriptsize $w$}, fill=utblue] (w2) 
        at (cw2) {};
        \node[vertex, fill=black] (u+v2) 
        at (cu+v2) {};
        \node[vertex, fill=black] (v12) 
        at (cv12) {};
        \node[vertex, fill=black] (v22) 
        at (cv22) {};
        \node[vertex, fill=black] (v32) 
        at (cv32) {};
        \node[vertex, fill=black] (u12) 
        at (cu12) {};
        \node[vertex, fill=black] (u22) 
        at (cu22) {};
        \node[vertex, fill=black] (u32) 
        at (cu32) {};
    
        \path[-,very thick] 
        (w2) edge (v12)
        (w2) edge (v22)
        (w2) edge (v32)
        (w2) edge[bend left] (u+v2)
        (w2) edge (u12)
        (w2) edge (u22)
        (w2) edge (u32)
        (w2) edge[bend right] (u+v2)
        (v22) edge (v12)
        (v22) edge (v32)
        (u12) edge (u22)
        (u12) edge (u+v2)
        ;
        \end{tikzpicture}
    \caption{The blue edge is contracted, combining vertices $u$ and $v$ into a single vertex $w$ and creating two parallel edges. }
    \label{fig:cont}
\end{figure}
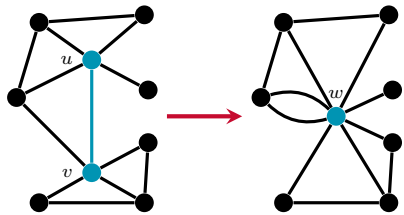

\section{Preparatory Results}
\label{sec:tech_lemmas}
We here collect a few basic observations about contractions, bonds and minimum spanning trees that we use in later proofs. The reader may choose to skip this section for now and continue with Section~\ref{sec:model}.

\begin{lemma}
	\label{lemma_bond_decrease}
	If $G=(V,E)$ is a connected (multi)graph, then for all 
 edges $e\in{}E$, we have $|B(G/e)|\leq{}|B(G)|$.
\end{lemma}
\begin{proof}
The proof is by contradiction.
Assume bond $B'$ in $G/e$, some maximum cardinality bond $B$ in $G$, and $|B'| > |B|$. As $B'$ cannot be a bond in $G$, it is either not a cut set, or not a minimal cut set in $G$. 
First, suppose $B'$ is not a cut set in $G$. But since $B' \cup{} \{e\}$ is a cut set in $G$, $B' \cup{} \{e\}$ must be a bond in $G$, and $|B|< |B' \cup \{e\}|$, a contradiction. Next consider the case that $B'$ is not a minimal cut set in $G$, i.e., it contains another cut set $C\subset B'$. Then $C$ is still a cut set in $G/e$, so $B'$ cannot be a bond in $G/e$, a contradiction. \qed
\end{proof}
\begin{lemma}
	\label{lemma_contract_bond}
	If $G=(V,E)$ is a connected (multi)graph with $|V|>2$, there exists some edge $e\in{}E$  so that $B(G) = B(G/e)$.
\end{lemma}
\begin{proof}
	By the definition of a bond, the removal of the edges of the bond increases the number of connected components in the graph by exactly one. Since we assumed $|V|>2$, at least one of the two connected components contains at least 2 vertices. Since $G$ was connected, this connected component contains an edge. Thus, there exists an edge $e\not\in B$, and $B$ remains a bond in $G/e$.\qed
\end{proof}

\begin{lemma}
	\label{lemma_edge_bound}
	Let $G=(V,E)$ be a connected (multi)graph without loops. Furthermore, let $T$ be a spanning tree in $G$, and let $B(G)$ be a largest bond in $G$. Then 
	\begin{equation}
		|E| \leq{} |B(G)|\cdot|T|\, .
	\end{equation}
\end{lemma}
\begin{proof}
The proof follows from the correctness of Prim's algorithm to compute minimum spanning trees \cite{Prim1957}. Prim's algorithm is initialized by partitioning $G$ into  $V_1 = \{v\}$ and $V_2 = V\setminus \{v\}$. Furthermore, an empty tree $T$ is initialized. Note that the edges which have one end-point in $V_1$ and one end-point in $V_2$ form a bond, say $B_0$. After initialization, an edge $e = \{u,v\}$ with minimal weight among the edges of $B_0$ is added to $T$, vertex $v$ is removed from $V_2$ and added to $V_1$. After this, $V_1$ and $V_2$ induce a new bond $B_1$, and this process is repeated until $V_2$ is empty.
Note that  $|B_i|\le |B(G)|$ in each iteration. Moreover, by correctness of Prim's algorithm we must have $E=B_0\cup B_2\cup\cdots\cup B_{|T|}$.\qed
\end{proof}

\begin{lemma}
\label{lemma_circuit_same_weight}
    Suppose there are two distinct MSTs $T_1$ and $T_2$ in $G$. Furthermore, for edge $e = \{u,v\}$ which is not contained in $T_1$ neither in $T_2$, let $f_i$ be a maximum weight edge on $T_i$'s $\{u,v\}$-path, for $i\in{}\{1,2\}$. Then  $w(f_1) = w(f_2)$. 
\end{lemma}
\begin{proof}
    Assume without loss of generality that $w(f_1) \leq{} w(f_2)$. 
    We introduce new weights $\overline{w}$, which we let $\overline{w}(f) = w(f)$\, $\forall f\in{}E\setminus e$ and we let $\overline{w}(e) = \min\{w(f_1),w(f_2)\}$. Since we assumed $w(f_1) \leq{} w(f_2)$, $\overline{w}(e) = w(f_1)$. Then $T_1$ is $\overline{w}$-minimal since exchanging $f_1$ with $e$ would not decrease the weight. If any other edge exchange would decrease the weight, this would be a contradiction to the assumption that $T_1$ is $w$-minimal.
    Since we have $\overline{w}(T_1) = \overline{w}(T_2)$, we know that also $T_2$ is $\overline{w}$-minimal. Therefore, there is no edge we can exchange to obtain a $\overline{w}$-cheaper tree, which implies $\overline{w}(e) \geq{} \overline{w}(f_2)$. Since $\overline{w}(e) = \overline{w}(f_1)$, we get $w(f_1)=\overline{w}(f_1)=\overline{w}(e)\geq{}\overline{w}(f_2)=w(f_2)$.\qed
\end{proof}

\section{Single sample minimum spanning tree problem}
\label{sec:model}
From this point onward, we assume that the edge weights $w:E\to\R$ are stochastic. That is, the weights are governed by a random variable $W=(W(e))_{e\in E}$, and the edge weight of edge $e$ is realized as $w(e)\sim W(e)$. We assume that all $W(e)$ are independent. In the single sample setting, the algorithm does not know $W(e)$, but is given a single sample $\sw{e}\sim W(e)$ for each edge $e\in E$. Based on a given vector of sampled weights $\sw{e_1},\dots,\sw{e_m}$, the algorithm has to compute a spanning tree $\talg$ with the aim to \emph{minimize the expected weight}
\[
    \mathbb{E}_{w\sim W}\left[{w(\talg)}\right]\,.
\]
Arguably, the only reasonable algorithm in this stochastic single sample setting is to choose a minimum spanning tree based on the sampled weights $\sw{e}$.
\begin{definition}[\alg]
    Denote by \alg\ an efficient implementation of Kruskal's algorithm to compute a spanning tree with minimal total weight for the sampled edge weights $\w$, so that for any $\w$, $T^{\alg}=\textup{argmin}_{T\in \trees} \sw{T}$.
\end{definition}
To gauge the quality of this sampling based algorithm, we use as benchmark the \emph{non-adaptive} adversary, which is the optimal solution when the distributions $W(e)$ are known, yet the adversary suffers from the same uncertainty as to realized weights $w(e)$. 
Our benchmark therefore is
\begin{align}\label{eq:opt}
     \opt:=\min_{T\in\trees}\mathbb{E}_{w\sim W}[w(T)]\,.
\end{align}
We refer to a corresponding tree that achieves the optimal solution as $T^{\opt}:=\text{argmin}_{T\in\trees}\mathbb{E}_{w\sim W}\left[{w}(T)\right]$.
It follows from 
linearity of expectation that this non-adaptive optimal solution $T^{\opt}$ is just a minimum spanning tree based on \emph{expected} edge weights, since for any tree $T$, $\mathbb{E}_{w\sim W}\left[{w}(T)\right]=\sum_{e\in T}\mathbb{E}_{w\sim W}\left[w(e)\right]$.

Our primary goal is to find conditions on the weight distributions $W(e)$ under which the sampling based algorithm \alg\ has a provably good expected performance, so that for some $\alpha\ge 1$
\begin{align*}
    \mathbb{E}_{\w\sim W}\mathbb{E}_{w\sim W}\left[w(\tsam)\right]\le \alpha \cdot \opt\,.
\end{align*}
Note the double expectation, because $\tsam$ depends on $\w$. If no ambiguity arises, we will henceforth simply write $\E{\cdot}$ instead. 

Benchmark \opt\ is comparable to the ones also being used in other stochastic combinatorial optimization problems, e.g., in stochastic scheduling~\cite{MSU99} or stochastic knapsack~\cite{DeanGoemansVondrak2008}. Alternatively, one could consider a stronger benchmark, namely the \emph{adaptive} optimum which is the expected value $\mathbb{E}_{w\sim W}\left[\min_{T\in\trees}w(T)\right]$. This expectation of the cost of the pointwise minimum spanning tree and the corresponding random variable has been analyzed in several papers~\cite{Kulkarni1988,Jain1988,HutsonShier2005}.
However, already \cite{Jain1988} gives examples to show a large gap between the values that we call non-adaptive and adaptive optimum here (cp.\ $E\xi$ and $g$ in \cite{Jain1988}).  Indeed, for illustration consider the following example, analogous to the one in~\cite{PuckUetz24}, that separates the two different adversaries and shows that 
in general, the adaptive adversary may be too strong for the single sample setting that we consider here: Take a graph with two vertices, two parallel edges, one edge with weight $1$, and the other edge with weight $0$ with probability $(1-1/M)$ and weight $M^2$ with probability $1/M$, for large integer $M$. Then $\mathbb{E}_{w\sim W}\left[\min_{T\in\trees}w(T)\right]=1/M$, while $\min_{T\in\trees}\mathbb{E}_{w\sim W}[w(T)]=1$.

Note that the same example shows that, without any  assumptions on the weight distributions $W(e)$, the performance of the sampling based algorithm \sam\ can be arbitrarily bad even in comparison to the non-adaptive adversary~\eqref{eq:opt}, because the samples are misleading with overwhelming probability: Here algorithm \sam\ chooses the wrong edge with probability $\approx 1$, hence the expected cost $\E{w(\tsam)}$ is yet another factor $\approx M$ higher than $\opt$. 

As a matter of fact, there are simple examples that rule out reasonable classes of candidate distributions, too, e.g., symmetric weight distributions: Take a graph with two vertices and two parallel edges, one with weight $1$, and the other with stochastic weight $0$ or $M>2$, each with probability $1/2$. Then $\opt{}=1$, while the expected cost of \sam{} equals $M/2$.
In the rest of the paper, we therefore restrict to exponentially distributed edge weights, and show that this suffices to obtain provably good performance guarantees. While it is true that this restriction is mainly motivated by its nice analytical properties, we note that we are not the first to study stochastic spanning trees in the setting with exponentially distributed edge weights, e.g.,~\cite{Kulkarni1988,Bailey1994}.

\section{Exponentially distributed edge weights}
\label{sec:main_results}
Here we give the proof of the main result of this paper. Recall that $B(G)$ denotes a largest bond in a given (multi-)graph $G=(V,E)$ with $|E|=m$, and let us define $b:=|B(G)|$ as the size of a largest bond of $G$. Moreover, for exponentially distributed edge weights, we denote by $\llamb_G=(\lambda_1,\dots,\lambda_m)$ the vector of rate parameters of the exponential distributions of edges, so that for edge $e\in E$ we have $w(e)\sim \exp(\lambda_e)$. Our main result follows.
\begin{theorem}
\label{main_thm_graphs}
	For all connected graphs $G$ with exponentially distributed stochastic edge weights, algorithm \sam\ is a $b$-approximation algorithm, that is, 
 \[
 \E{w(\tsam)}\le b\cdot\textup{\opt}\,.
 \]
 Furthermore, for all connected graphs $G$ and all $\eps>0$, there exist exponentially distributed edge weights for $G$ such that $\E{w(\tsam)}\ge (b-\eps)\cdot\textup{\opt}$.
\end{theorem}
Thus, not only is the performance of the sampling based algorithm \sam{} purely determined by the structure of the graph~$G$, but for \emph{any} given graph $G$ we can also find a rate vector $\llamb_G$ such that this bound is asymptotically tight.

The proof of the theorem goes by two independent inductive arguments via edge contractions. In order to formulate it, we first introduce a bit of shorthand notation. 
To that end, recall that an instance of the single sample stochastic minimum spanning tree problem is defined by the pair $(G,\llamb_G)$, where $\llamb_G=(\lambda_1,\dots,\lambda_m)$ is the vector of rate parameters of $G$'s edges. Let $H$ be a graph which can be obtained from $G$ by contracting edges. Then by $\llamb_H$ we denote the rate vector $\llamb_G$ restricted to edges of graph $H$, that is, $\llamb_H = (\lambda_e)_{e\in{}E(H)}$, where $E(H)\subset E(G)$ denotes the edge set of graph $H$. 

Since we will need to argue about the expected costs of both algorithm and optimum on graphs $G$ and $G/e$, and also sometimes  manipulate rate parameters $\llamb_G$, it is convenient to use the additional notation
\[
    \E{\alg(G,\llamb_G)}\quad\text{as well as}\quad \E{\opt(G,\llamb_G)}
\]
for the expected cost of \sam\ and the optimal solution $T^{\opt}$, respectively on instance $(G,\llamb_G)$. Moreover, lets us denote the relative performance of \sam\  on a given instance $(G,\llamb_G)$ by
\[
    \alpha(G,\llamb_G):= \E{\alg(G,\llamb_G)} / \E{\opt(G,\llamb_G)}\,.
\]
Next, by $\ch{G,e,j}$ we denote the event that the sampling based algorithm \sam{} chooses edge $e$ as its $j$th edge  to be included in $\tsam$, on a given instance $(G,\llamb_G)$. For $j=1$, this denotes the event where edge $e$ is the first edge chosen by \sam{} (which can happen only if $\sw{e}\le \sw{f}$ for all edges $f\in E$).

Finally if the first edge chosen by the sampling based algorithm \sam{}, say edge $e$, is \emph{not} part of any optimal solution $T^{\opt}$, $T^{\opt}\cup{}e$ contains a fundamental cycle, $C(T\cup\{e\})$. Then, a spanning tree with minimal expected weight \emph{conditioned on edge $e$ being included} is obtained by exchanging edge $e$ with an edge with the second largest expected weight among the edges in $C(T\cup\{e\})$. Given $e$ and $T^{\opt}$, denote such an edge by $e^*$. Lemma~\ref{lemma_circuit_same_weight} now yields that the \emph{expected weight} $\E{w(e^*)}=1/\lambda_{e^*}$ is in fact independent of the choice of the tree $T^{\opt}$.
In the case where $e\in{}T^{\opt}$, we simply define $e^* := e$.

\subsection{Auxiliary results for edge contractions}
First, as basis for the subsequent inductive proof, a result about the expected cost of \sam\ conditioned on a particular edge being chosen first. This result requires the edge weights to be memoryless.
\begin{lemma}
	\label{lemma_conditional_sam}
 For every graph $G$ with rate parameters $\llamb_G$ and every edge $e\in E$, 
	\begin{equation*}
		\mathbb{E}[\sam(G,\llamb_G)|\ch{G,e,1}] -  {1}/{\lambda_{e}} = \mathbb{E}[\sam(G/e,\llamb_{G/e})]\,. 
	\end{equation*}
\end{lemma}
\begin{proof}
	The proof follows because: (i) conditioned on $\ch{G,e,1}$ with $e=\{u,v\}$, the spanning tree $\tsam$ cannot contain any other $\{u,v\}$-path, so $\tsam/e$ is a spanning tree in $G/e$. Moreover, (ii) a spanning tree $T'$ in $G/e$  yields a spanning tree in $G$, namely $T'\cup \{e\}$.  Hence for each realization of samples $\w$, given the condition, the claim holds. Moreover, by the fact that edge weights are exponentially distributed and hence memoryless, the condition of $e$ being chosen first in $G$ does no harm, because the probabilities of \sam\ choosing edges in $G\setminus e$, given the condition, and \sam\ choosing edges in $G/e$ are identical. Formally, after the first edge $e$ has been chosen, this reads as $\mathbb{P}(\textsc{Ch}(G,f,2)|\textsc{Ch}(G,e,1)) = \mathbb{P}(\textsc{Ch}(G/e,f,1))$
    for all $f\neq e$. Therefore the claim also holds in expectation. \qed
\end{proof}

Next we derive a comparable result for \opt. To that end, given $e\in E$ and some $T^{\opt}$, recall that we defined $e^*=e$ if $e\in{}T^{\opt}$, and otherwise $e^*$ denotes an edge with the second largest expected weight in the unique cycle $C\subseteq T^{\opt}\cup \{e\}$.
\begin{lemma}
	\label{lemma_contract_opt}
 For every graph $G$ with rate parameters $\llamb_G$ and every edge $e\in E$, and $e^*$ (depending on $e$) as defined above, 
	\begin{equation*}
		 \E{\opt(G,\llamb_G)} - 1/\lambda_{e^*} = \E{\opt(G/e,\llamb_{G/e})}\,.
	\end{equation*}
\end{lemma}

\begin{proof}
	\textbf{Case 1.} (There exists $T^{\opt}$ with $e\in{}T^{\opt}$)
 Since $e\in{}T^{\opt}$, we have $e^* = e$. Furthermore, $T^{\opt}/e$ is an MST in $G/e$ with respect to expected weights. If not, there exists some other tree $T'$ in $G/e$ which has a smaller expected total weight. However, this would imply that $T'\cup{}e$ is also cheaper in expectation than $T^{\opt}$, a contradiction.  This gives
\begin{align*}
		\mathbb{E}[\opt(G/e,\llamb_{G/e})]  
  &= \mathbb{E}[\opt(G,\llamb_G)] - 1/\lambda_e = \mathbb{E}[\opt(G,\llamb_G)] - 1/\lambda_{e^*}\,.
	\end{align*}
	\textbf{Case 2.} (For all $T^{\opt}$, $e\notin{}T^{\opt}$)
  Given some $T^{\opt}$, to obtain a minimum spanning tree conditioned on including edge $e$, we can exchange $e^*$ with $e$. We know by Lemma~\ref{lemma_circuit_same_weight} that the expected weight of $e^*$ is independent of the choice of $T^{\opt}$, so that we can write
 \[
 \E{\opt(G,\llamb_G)|e\in T^{\opt}}= \E{\opt(G,\llamb_G)}+1/\lambda_e-1/\lambda_{e^*}\,.
 \]
 Moreover, by the same line of argument as in Case~1, we have
 \[
\E{\opt(G/e,\llamb_{G/e})}+ 1/\lambda_e = \E{\opt(G,\llamb_G)|e\in T^{\opt}}\,,
 \]
 and the claim follows.\qed
\end{proof}

\subsection{Upper bound on the performance of sampling based algorithm}
With the two previous auxiliary results, we can prove an upper bound on \sam's performance as follows. To that end, denote by $\Lambda_G$ the set of all non-negative rate parameters $\llamb_G$.
\begin{lemma}
	\label{upper_bound_thm}
		For all connected (multi)graphs $G$ with exponentially distributed edge weights, the performance bound of the sampling based algorithm \sam{} is no worse than the size of the largest bond in $G$. That is,
	\begin{equation*}
	\sup_{\llamb_G\in{}\Lambda_G} \alpha(G,\llamb_G) \leq{} |B(G)|\,.
\end{equation*}
\end{lemma}

\begin{proof}
The proof is by induction on the number of edges of $G$. First, the lemma clearly holds for connected graphs with $m=1$ edges. So assume that the lemma holds for all (multi)graphs with $m=1$ up to $k$ edges, $k\ge 1$. 
Consider connected graph $G$ with $k+1\ge 2$ edges. 
Then for all $\llamb_G\in\Lambda_G$:
\begin{align}\nonumber
 \alpha(G,\llamb_G) &= 
 \frac{\mathbb{E}[\sam(G,\llamb_G)]}{\mathbb{E}[\opt(G,\llamb_G)]}\\\nonumber
	&= 
  \frac{\sum_{e\in{}E}\mathbb{P}(\textsc{Ch}(G,e,1))\mathbb{E}[\sam(G,\llamb_G)|\textsc{Ch}(G,e,1)]}{\mathbb{E}[\opt(G,\llamb_G)]}\\
	&=  
 \frac{\sum_{e\in{}E}\mathbb{P}(\textsc{Ch}(G,e,1))\Bigl(\mathbb{E}[\sam(G/e,\llamb_{G/e})]+1/\lambda_e\Bigr)}{\mathbb{E}[\opt(G,\llamb_G)]}\,.\label{eq:proof1}
        \end{align}
Here, we have first conditioned on which edge is first chosen by \sam, and the second equality holds by Lemma~\ref{lemma_conditional_sam}. Now by the definition of $\alpha$, \eqref{eq:proof1} equals
\begin{align}\label{eq:proof2}
	& 
        \frac{\sum_{e\in{}E}\mathbb{P}(\textsc{Ch}(G,e,1))\Bigl(\alpha(G/e,\llamb_{G/e})\mathbb{E}[\opt(G/e,\llamb_{G/e})]+1/\lambda_e\Bigr)}{\mathbb{E}[\opt(G,\llamb_G)]}\,.
\end{align}
Now by the induction hypothesis  $\alpha(G/e,\llamb_{G/e}) \leq{} |B(G/e)|$.  Also, by Lemma~\ref{lemma_bond_decrease} we have $|B(G/e)|\le |B(G)|$.
So via \eqref{eq:proof2}, $\alpha(G,\llamb_G)$ is upper bounded by
\begin{align}\label{eq:proof3}
	& |B(G)|\frac{\sum_{e\in{}E}\mathbb{P}(\textsc{Ch}(G,e,1))\Bigl(\mathbb{E}[\opt(G/e,\llamb_{G/e})]+{1}/(\lambda_e|B(G)|)\Bigr)}{\mathbb{E}[\opt(G,\llamb_G)]}\,.
 \end{align}
Now if we apply Lemma \ref{lemma_contract_opt}, \eqref{eq:proof3} equals
\begin{align*}
	& |B(G)|\frac{\mathbb{E}[\opt(G,\llamb_G)]+\sum_{e\in{}E}\mathbb{P}(\textsc{Ch}(G,e,1))\Bigl({1}/(\lambda_e|B(G)|)-1/\lambda_{e^*}\Bigr)}{\mathbb{E}[\opt(G,\llamb_G)]}\,.
 \end{align*}
That means that the proof is done if it is true that for all $\llamb_G\in\Lambda_G$
\begin{align}\label{eq:proof4}
    \sum_{e\in{}E}\mathbb{P}(\textsc{Ch}(G,e,1))\Bigl({1}/(\lambda_e|B(G)|)-1/\lambda_{e^*}\Bigr)\le 0\,.
\end{align}
With the help of Lemma~\ref{lemma_edge_bound}, we can indeed show  that this is the case; see Appendix~\ref{sec:missingproofs}.
Thus, we conclude that $\alpha(G,\llamb_G) \leq{} |B(G)|$ for all $\llamb_G\in\Lambda_G$. \qed
\end{proof}

\subsection{Lower bound on the performance of sampling based algorithm}
Maybe surprisingly, the upper bound is simultaneously a lower bound for any given graph $G$. This is achieved by choosing corresponding rate distributions $\llamb_G$ for given $G$.  The following lemma completes the proof of Theorem~\ref{main_thm_graphs}.
\begin{lemma}
	\label{lower_bound_thm}
	For all connected (multi)graphs $G$ with exponentially distributed edge weights, the performance bound of the sampling based algorithm \sam{} can be no better than $|B(G)|$. That is,
	\begin{equation*}
	\sup_{\llamb_G\in{}\Lambda_G} \alpha(G,\llamb_G) \geq{} |B(G)| \,\text{.}
\end{equation*}
\end{lemma}
\begin{proof}
The proof goes again  by induction on $m$, the number of edges in $G$. For the case $m=1$, it is clear that there is only one spanning tree to choose from, thus the performance bound of \sam{} equals 1, which is (greater than or) equal to $|B(G)|$. Assume the statement holds for graphs with up to $k$ edges, and consider a graph $G$ with $m=k+1\ge 2$ edges, so in particular $|V|>2$. By subsequently using conditioning on the first edge being chosen by \sam, and then applying in this order Lemma~\ref{lemma_conditional_sam}, the definition of $\alpha$, and Lemma~\ref{lemma_contract_opt}, we get that 
\begin{align*} 
  \alpha(G,\llamb_G) &= 
  \frac{\sum_{e\in{}E}\mathbb{P}(\textsc{Ch}(G,e,1))\mathbb{E}[\sam(G,\llamb_G)|\textsc{Ch}(G,e,1)]}{\mathbb{E}[\opt(G,\llamb_G)]}\\
&= 
\frac{\sum_{e\in{}E}\mathbb{P}(\textsc{Ch}(G,e,1))\Bigl(\mathbb{E}[\sam(G/e,\llamb_{G/e})]+{1}/{\lambda_e}\Bigr)}{\mathbb{E}[\opt(G,\llamb_G)]}\\
&= 
\frac{\sum_{e\in{}E}\mathbb{P}(\textsc{Ch}(G,e,1))\Bigl(\alpha(G/e,\llamb_{G/e})\mathbb{E}[\opt(G/e,\llamb_{G/e})]+1/\lambda_e\Bigr)}{\mathbb{E}[\opt(G,\llamb_G)]}\\
&=
\sum_{e\in{}E}\mathbb{P}(\textsc{Ch}(G,e,1))\alpha(G/e,\llamb_{G/e})\\
        &\hspace{1cm}+\sum_{e\in{}E}\mathbb{P}(\textsc{Ch}(G,e,1))\left(\frac{1}{\lambda_e\mathbb{E}[\opt(G,\llamb_G)]}-\frac{\alpha(G/e,\llamb_{G/e})}{\lambda_{e^*}\mathbb{E}[\opt(G,\llamb_G)]}\right)\,.
\end{align*}
Now choose an edge $\ove \notin{} B(G)$, which by  Lemma~\ref{lemma_contract_bond} exists because $G$ is connected and $|V|>2$.  The idea is to let the rate $\lambda_{\ove}$ tend to infinity, while assuming $\lambda_e\in{}o(\lambda_{\overline{e}})\, \forall{}e\neq{}\overline{e}$, so that in this limit, algorithm \sam\ chooses edge $\ove$ with probability 1. Note that this also implies that $e^*=\ove$, because edge $\ove$ is the unique cheapest edge, which implies $\ove\in T^{\opt}$. Then, taking the limit $\lambda_{\ove}\to\infty$ makes all but the $\ove$-term vanish in the above equation for $\alpha(G,\llamb_G)$, since all terms in the above equation are bounded for fixed $G$. This therefore yields 
\begin{align*}
\sup_{\llamb_G\in{}\Lambda_G} \alpha(G,\llamb_G) &\geq{}
\sup_{\llamb_G\in{}\Lambda_G,\lambda_{\ove}\to\infty} \Bigg(\alpha(G/\ove,\llamb_{G/\ove}) +\frac{1-\alpha(G/\ove,\llamb_{G/\ove})}{\lambda_{\ove}\mathbb{E}[\opt(G,\llamb_G)]}\Bigg)\\
         &= \sup_{\llamb_G \in{} \Lambda_{G/\ove}}\alpha(G/\ove,\llamb_{G/\ove})\,,
\end{align*}
since the first term does not depend on $\lambda_{\ove}$ and last term vanishes. Now we can use the 
induction hypothesis, which gives
\begin{align*}
	\sup_{\llamb_G \in{} \Lambda_{G/\ove}}\alpha(G/\overline{e},\llamb_{G/\ove})&\geq{} |B(G/\overline{e})|\,.
\end{align*}
Since we had chosen $\overline{e}\notin{}B(G)$, we finally have $|B(G/\overline{e})| = |B(G)|$ by Lemma~\ref{lemma_contract_bond}.
Thus, we have shown that $\sup_{\llamb_G\in{}\Lambda_G} \alpha(G,\llamb_G) \geq{} |B(G)|$.\qed
\end{proof}

\subsection{Interpretation as item selection}
At first sight it may seem surprising that the tight performance bound
does only depend on the structure of the graph. However since our proof is by induction and using edge contractions, there is a simple explanation of the result: By subsequent edge contractions, one can reach a graph with only two vertices and a number of parallel edges, while that number is equal to the largest bond in the graph. The spanning tree problem on such a two-vertex (multi)graph is equivalent to selecting the lightest of a number of items, each with a stochastic weight, while the selection must be done on the basis of only one sample per item. 
Hence, one can interpret our inductive proof also as a reduction of the minimum spanning tree problem to an  stochastic minimum cost single item selection problem.  The tight performance bound for that problem is equal to the number of items.

\section{Extension to matroids}
\label{sec:main_matroids}
One may wonder if other combinatorial subset selection problems with stochastic costs are amenable to the same type of analysis as proposed in this paper.  For our techniques to work, concepts such as edge contractions and bonds should then have a well defined meaning and correspondence.
That being said, a generalization from spanning trees to matroids $M$ on ground set of elements $E$ suggests itself. Indeed, when the elements $E$ of a matroid $M=(E,\ind)$ have exponentially distributed weights, this problem has earlier been considered, e.g., in~\cite{Bailey1994} for analyzing 
the distribution function of pointwise minimum weight bases.

For matroids, edge contractions translate into the contraction of the matroid with respect to singleton elements $e\in E$:  The contraction $M/e$ of a loopless matroid $M$ is a matroid on ground set $M\setminus\{e\}$,  with independent sets $F\subseteq M\setminus \{e\}$ such that $F\cup\{e\}$ is independent in $M$. Moreover, the largest bond of a graph translates into a largest cocircuit of the matroid $M$.
Finally, Kruskal's greedy algorithm~\cite{Kruskal1956} based on sampled edge weights for minimum spanning trees has its correspondence in Edmonds' greedy algorithm for weighted matroids~\cite{Edmonds1971}, so that \sam{} would be the greedy algorithm applied to the matroid with sampled element weights. That being said, one can easily verify that the same inductive proof that we use to derive 
Lemma~\ref{upper_bound_thm}, as well as Lemmas~\ref{lemma_conditional_sam} and~\ref{lemma_contract_opt}, generalize one-to-one to matroids by carefully replacing the corresponding terms, that is, edge $\to$ element, minimum spanning tree $\to$ minimum weight basis, etc. The basis of the inductive proof of Lemma~\ref{upper_bound_thm} generalizes directly, because \emph{all} matroids with at most three elements are known to be graphic \cite[p.\ 12]{Oxley}. 
This yields the following.
\begin{theorem}\label{thm:matroids}
    For loopless matroid $M$ with exponentially distributed element weights, the performance bound of the sampling based greedy algorithm \sam{} is no worse than $c^*(M)$, the size  of a largest cocircuit of $M$. 
\end{theorem}
As we argued before in Lemma~\ref{lower_bound_thm}, for graphical matroids this result is also tight. What still needs to be verified is the generalizations of Lemma~\ref{lemma_bond_decrease},~\ref{lemma_edge_bound}, and~\ref{lemma_circuit_same_weight} to matroids, as these are necessary to derive the upper bound in Lemma~\ref{upper_bound_thm}. For this, we refer to Appendix~\ref{sec:matroids}.

\section{Conclusions}
It is a natural question if the results that we give can be extended to weight distributions other than exponentials. In the  proof of Lemma~\ref{lemma_conditional_sam} as well as in the proof of \eqref{eq:proof4} we require the distributions to be memoryless. Noting that $\le$ would suffice in Lemma~\ref{lemma_conditional_sam} to derive an upper bound on \sam's performance, it is not clear how that would allow to generalize to a class broader than exponentials.
Furthermore, and in line with the prophet inequality setting, it is tempting to consider maximization instead of minimization problems. Then the effect of choosing a ``wrong'' edge (or item) is less detrimental than it is for minimization problems. Small examples suggest a much stronger performance bound. It turns out, however, that the analysis is more challenging; we leave this for future work.

\subsection*{Acknowledgements} Thanks to Matthias Walter for suggesting the simple proof of Lemma~\ref{lemma_circuit_same_weight}, and to Jorn van der Pol for the short proof of Lemma~\ref{lem:matroid_edge_bound}.

\bibliographystyle{splncs04}
\bibliography{MST-One-Sample}

\appendix

\section{Omitted proofs}
\label{sec:missingproofs}

\subsection*{Proof of \eqref{eq:proof4}}
We have to show that for all $\llamb_G\in\Lambda_G$
\begin{align*}
    \sum_{e\in{}E}\mathbb{P}(\textsc{Ch}(G,e,1))\Bigl({1}/(\lambda_e|B(G)|)-1/\lambda_{e^*}\Bigr)\le 0\,.
\end{align*}
\begin{proof}
As edge weights of $G=(V,E)$ are exponentially distributed, we have that 
$\mathbb{P}(\textsc{Ch}(G,e,1))=\lambda_e/\sum_{f\in E} \lambda_f$. So, we need to show that
\[
\Big(1/\sum_{f\in{}E}\lambda_f\Big)\left[\sum_{e\in{}E}\left(\frac{1}{|B(G)|} - \frac{\lambda_e}{\lambda_{e^*}}\right)\right] \leq{} 0\,.
\]
This is equivalent to showing that
\begin{align*}
	\sum_{e\in{}E}\left(\frac{1}{|B(G)|} - \frac{\lambda_e}{\lambda_{e^*}}\right) \leq{} 0 \,.
\end{align*}
\end{proof}
To show this, we split the sum into edges in (some) $T^{\opt}$, and outside $T^{\opt}$. 
For~$e\in{}T^{\opt}$, we have $\lambda_e = \lambda_{e^*}$, thus,

\begin{align*}
	 \sum_{e\in{}E}\left(\frac{1}{|B(G)|} - \frac{\lambda_e}{\lambda_{e^*}}\right)&= \frac{|T^{\opt}|}{|B(G)|} - |T^{\opt}| + \sum_{e\notin{}T^{\opt}}\left(\frac{1}{|B(G)|} - \frac{\lambda_e}{\lambda_{e^*}}\right)\\
	 &\leq{} \frac{|T^{\opt}|}{|B(G)|} - |T^{\opt}| + \sum_{e\notin{}T^{\opt}}\frac{1}{|B(G)|}\,.\\
	 \intertext{By Lemma \ref{lemma_edge_bound}, we have $|E|\leq{} |B(G)|\cdot |T^{\opt}|$. Therefore,}
	\sum_{e\in{}E}\left(\frac{1}{|B(G)|} - \frac{\lambda_e}{\lambda_{e^*}}\right)&\leq{} \frac{|T^{\opt}|}{|B(G)|} - |T^{\opt}| + \frac{|B(G)||T^{\opt}|-|T^{\opt}|}{|B(G)|}=0
\end{align*}\qed

\section{Extension to Matroids}
\label{sec:matroids}

Here we give the missing lemmas to prove Theorem~\ref{thm:matroids}.  
We first collect a few  definitions and  facts about matroids. Refer to \cite{Oxley} for proofs and more details.
\begin{definition}[Matroid]
    Let $E$ be a set of elements and $\ind\subseteq 2^E$ a family of subsets of $E$. Then $M=(E,\ind)$ is a matroid if 
    \begin{enumerate}
        \item $\emptyset\in\ind$
        \item If $J\in \ind$ then $I\in\ind$ for all $I\subseteq J$ 
        \item If $I,J\in \ind$ and $|I|<|J|$, then there exists $e\in J\setminus I$ so that $I\cup\{e\}\in\ind$
    \end{enumerate}
\end{definition}
The sets in \ind\ are referred to as \emph{independent} sets, and sets outside \ind\ are \emph{dependent}.  The inclusion-wise maximal independent sets contained in $S\subseteq E$ are the bases of $S$. They all have the same size, denoted the \emph{rank} $r(S)$ of $S$. The bases of the matroid $M$ are the maximal independent subsets $B\subseteq E$, their size $|B|$ is the rank of the matroid, denoted $r(M)$. Minimal dependent sets $C\subseteq E$ are called \emph{circuits}. The graphical matroid of a connected graph $G=(V,E)$ is obtained via $\ind:=\{F\subseteq E\ |\ F\ \text{does not contain a cycle}\}$. It has rank $|V|-1$. If $M$ is a matroid with set of bases $\mathcal{B}(M)$, the \emph{dual} matroid $M^*$ is obtained from $M$ by taking $\mathcal{B}^*(M)=\{E\setminus B\ |\ B\in\mathcal{B}\}$ as its set of bases. (It is nontrivial to show that this indeed defines the bases of a matroid.) The circuits of $M^*$ are called the \emph{cocircuits} of $M$.  Cocircuits of $M$ are minimal subsets of $E$ that have a nontrivial intersection with every basis of $M$. If $M(G)$ is the graphical matroid for graph $G$, then the cocircuits of $M(G)$ are precisely the bonds of $G$ \cite[Prop.\ 2.3.1]{Oxley}.
For $S\subseteq E$, the \emph{closure} of $S$ is $cl(S)=\{e\in E\ |\  r(S\cup\{e\})=r(S)\}$. If $S=cl(S)$, then $S$ is called a \emph{flat}, and a flat $S$ with $|S|=r(M)-1$ is called a \emph{hyperplane} of $M$. It is well known that $X\subseteq E$ is a hyperplane of $M$ iff $E-X$ is a cocircuit of M~\cite[Prop.\ 2.1.6]{Oxley}.  A \emph{loop} of a matroid is a circuit consisting of a singleton element,~$\{e\}$. As with graphs, loops are irrelevant for computing minimum weight bases, so whenever $M$ or $M/e$ contains loops, we may eliminate them.
All this being defined, we get the following.

\begin{lemma}[Generalization of \protect{Lemma~\ref{lemma_bond_decrease}}]
    If $M=(E,\ind)$ is a matroid, and $c^*(M)$ the size of a largest cocircuit, then for all  $e\in{}E$, $c^*(M/e)\leq{}c^*(M)$.
\end{lemma}
\begin{proof}
    Assume there exists $e\in E$ and a cocircuit $C$ of $M/e$ that falsifies the claim. Then  $C$ is not a cocircuit in $M$, as otherwise  $|C|>c^*(M)$. That implies either that $C$ is not minimal while having a nonempty intersection with every basis of $M$, but then $C$ is not minimal with this property in $M/e$ either, because $e\not\in C$. Or it implies that
    there exists a basis $B$ of $M$ with $C\cap B=\emptyset$, but then  $C\cap (B\setminus\{e\})=\emptyset$, too, and $C$ is not a cocircuit of $M/e$. 
\qed\end{proof}

\begin{lemma}[Generalization of \protect{Lemma~\ref{lemma_edge_bound}}, see also \protect{\cite[Ex.\ 8 on p.133]{Oxley}}]\label{lem:matroid_edge_bound}
   For loopless matroid $M=(E,\ind)$, $|E| \le r(M)\cdot c^*(M)$, where $c^*(M)$ is the maximum size of a cocircuit of $M$.
\end{lemma}
\begin{proof}
    Let $B = \{b_1, b_2,\dots , b_r\}$ be a basis of $M$. For each $1 \le i \le r$, let $D_i := E(M)-cl(B-{b_i})$, which is the cocircuit of $M$ obtained by taking the complement of the hyperplane spanned by $B-\{b_i\}$. Note that $b_i\in D_i$. Then we claim that $E=\cup_{i=1}^r D_i$, which proves the claim, as each cocircuit $D_i$ has size at most $c^*(M)$. To see the claim, let $e\in E$, and let $C(e)$ be the unique circuit contained in $B\cup\{e\}$. As $M$ is loopless, $\{e\}\varsubsetneq C(e)$.  Then $e\in D_i$ for each $i$ with $b_i\in C(e)$, because $e\not\in D_i$ implies $e\in cl(B-b_i)$, contradicting the fact that $e$ and $b_i$ are both contained in $C(e)$.\qed 
\end{proof}

\begin{lemma}[Generalization of \protect{Lemma~\ref{lemma_circuit_same_weight}}]
    If $B_1$ and $B_2$ are two different minimum weight bases of matroid $M$, and $e\not\in B_1,B_2$, then the maximum weight elements in the unique circuits induced by $e$, say $f_1\in C(B_1\cup\{e\})$ and $f_2\in C(B_2\cup\{e\})$, have the same weight $w(f_1)=w(f_2)$.
\end{lemma}
\begin{proof}
    Assume w.l.o.g.\ that $w(f_1)\le w(f_2)$. We show 
    $w(f_1)\ge w(f_2)$.
    Just like in the proof of Lemma~\ref{lemma_circuit_same_weight}, define weight function $\overline{w}$ that is identical to $w$ but has
    $\overline{w}(e)=\min\{w(f_1),w(f_2)\}$, so that $\overline{w}(e)=w(f_1)$. Then $B_1$ is still a $\overline{w}$-minimal weight basis, because by definition of $\overline{w}$, $e$ has $\overline{w}$-maximal weight in the circuit $C(B_1\cup\{e\})$. Then also $B_2$ must be a $\overline{w}$-minimal basis, since $\overline{w}(B_1)=\overline{w}(B_2)$. This again implies $\overline{w}(e)\ge \overline{w}(f_2)$, because otherwise $\overline{w}((B_2\setminus\{f_2\})\cup\{e\})<\overline{w}(B_2)$. Hence $w(f_1)=\overline{w}(f_1)=\overline{w}(e)\ge \overline{w}(f_2)=w(f_2)$.
\qed\end{proof}
\end{document}